\documentclass{IEEEtran}

\usepackage{tikz}
\usetikzlibrary{shapes,arrows}
\usepackage{hyperref}

\usepackage{amsmath}
\usepackage{amssymb}
\usepackage{amsthm}
\usepackage{bm} 
\usepackage{booktabs} 
\usepackage{multirow}
\usepackage{xspace}

\usepackage{cite}
\usepackage{algorithm, algorithmic}

\foreach \x in {a,...,z}{
  \expandafter\xdef\csname bf\x \endcsname{\noexpand\ensuremath{\noexpand\mathbf{\x}}}
  \expandafter\xdef\csname bm\x \endcsname{\noexpand\ensuremath{\noexpand\boldsymbol{\x}}}
}

\foreach \x in {A,...,Z}{
  \expandafter\xdef\csname bs\x \endcsname{\noexpand\ensuremath{\noexpand\boldsymbol{\x}}}
  \expandafter\xdef\csname bf\x \endcsname{\noexpand\ensuremath{\noexpand\mathbf{\x}}}
  \expandafter\xdef\csname bb\x \endcsname{\noexpand\ensuremath{\noexpand\mathbb{\x}}}
  \expandafter\xdef\csname ds\x \endcsname{\noexpand\ensuremath{\noexpand\mathds{\x}}}
  \expandafter\xdef\csname cal\x \endcsname{\noexpand\ensuremath{\noexpand\mathcal{\x}}}
}

\DeclareMathOperator{\trace}{tr}

\def\transp{\top}

\def\bmmu{\bm{\mu}}
\renewcommand{\phi}{\varphi}

\def\bbHS{\bbH_{\calS}}

\def\real{\mathrm{Re}\,}
\def\imag{\mathrm{Im}\,}
\def\imagi{\mathrm{Im}_{\bmi}\,}
\def\imagj{\mathrm{Im}_{\bmj}\,}
\def\imagk{\mathrm{Im}_{\bmk}\,}

\newcommand{\ip}[1]{\left\langle #1 \right\rangle}

\renewcommand{\epsilon}{\varepsilon}

\newcommand{\matrixspace}[3]{{#1}^{#2 \times #3}}
\renewcommand{\tilde}{\widetilde}

\def\ie{\textit{i.e.}\xspace}

\def\argmin{\operatornamewithlimits{\arg\min}}
\def\bbHS{\bbH_{\calS}}
\def\intnmf{\calI^{\mathrm{NMF}}}
\def\intqnmf{\calI^{\mathrm{QNMF}}}
\def\intpol{\calI^{\mathrm{pol}}}

\newtheorem{proposition}{Proposition}

\newtheorem{definition}{Definition}

\begin{document}

\title{Quaternion Non-negative Matrix Factorization: \\ definition, uniqueness and algorithm}

\author{Julien Flamant, Sebastian Miron, David Brie
\thanks{All authors are with Université de Lorraine, CNRS, CRAN, F-54000 Nancy, France. This work has been funded by Région Grand-Est and project \textsc{ANR-15-CE10-0007 optifin}. Part of the material in this paper was presented at the 8th International Workshop on Computational Advances in Multi-Sensor Adaptive Processing (CAMSAP), December 15–18, 2019 \cite{flamant2019qnmf}.}}


\maketitle


\begin{abstract}
This article introduces quaternion non-negative matrix factorization (QNMF), which generalizes the usual non-negative matrix factorization (NMF) to the case of polarized signals.
Polarization information is represented by Stokes parameters, a set of 4 energetic parameters widely used in polarimetric imaging.
QNMF relies on two key ingredients: \emph{(i)} the algebraic representation of Stokes parameters thanks to quaternions and \emph{(ii)}
the exploitation of physical constraints on Stokes parameters.
These constraints generalize non-negativity to the case of polarized signals, encoding positive semi-definiteness of the covariance matrix associated which each source.
Uniqueness conditions for QNMF are presented.
Remarkably, they encompass known sufficient uniqueness conditions from NMF.
Meanwhile, QNMF further relaxes NMF uniqueness conditions requiring sources to exhibit a certain zero-pattern, by leveraging the complete polarization information.
We introduce a simple yet efficient algorithm called quaternion alternating least squares (QALS) to solve the QNMF problem in practice.
Closed-form quaternion updates are obtained using the recently introduced generalized HR calculus.
Numerical experiments on synthetic data demonstrate the relevance of the approach.
QNMF defines a promising generic low-rank approximation tool to handle polarization, notably for blind source separation problems arising in imaging applications.
\end{abstract}


\section{Introduction}

\IEEEPARstart{P}{olarization} information is essential to many fields ranging from seismology \cite{akiRichards2002}, optics \cite{born2000principles} to gravitational astronomy \cite{misner1973gravitation}, among others.
Polarization encodes the geometry of wave oscillations.
It carries crucial morphological and physical insights about the medium which transmitted or reflected the polarized wave.
For instance, polarization permits the detection of hidden targets in remote sensing \cite{lee2009polarimetric,tyo2006review}, it provides insights into solar magnetic fields \cite{orozco2007quiet} and reveals the anisotropy of biological tissues \cite{ghosh2011tissue}.
In summary, \emph{polarimetric imaging} proves to be useful in a wide range of situations as it provides unique features of the observed scene, unaccessible to conventional intensity imaging.

Recent years have seen an increased interest in exploiting polarization diversity in hyperspectral imaging systems \cite{mu2012static,zhao2009spectropolarimetric, miron2011joint}.
This modality, called \emph{spectro-polarimetric imaging}, was first popularized in astronomy \cite{antonucci1984optical}.
In full generality, it consists in acquiring, for a collection of wavelengths and locations, a 4-dimensional real vector gathering the 4 Stokes parameters.
These energetic parameters are widely used in optics \cite{J.J.Gil2016} to describe the complete properties of light (intensity and polarization).
The joint acquisition of spatial, spectral and polarization diversities raises important and original challenges, in particular regarding blind polarized source separation problems.

Spectro-polarimetric data can be represented as a third order real-valued tensor, where the rows, columns and fibers correspond to frequency, space and polarization diversities, respectively.
Note that, since polarization is described here by Stokes parameters, the fiber is of size 4.
Solving the blind polarized source separation problem using usual low-rank approximation techniques raises two major issues.
First, Stokes parameters being energetic quantities, they obey structural constraints \cite{J.J.Gil2016} which need to be taken into account to guarantee the physical interpretation of the recovered rank-1 terms.
Second, low-rank decompositions of spectro-polarimetric images such as the Canonical Polyadic Decomposition (CPD) \cite{kolda2009tensor} necessarily yield rank-1 terms with \emph{constant polarization}.
It thus restricts the physical validity of usual tensor decompositions to specific settings, \emph{e.g.} narrow band sources with constant polarization.
The general case of wideband sources with frequency-dependent polarization, as observed in spectro-polarimetric data, requires the development of new low-rank approximation tools to deal with the specificities of polarization.

To this aim, this paper introduces quaternion non-negative matrix factorization (QNMF).
It extends the concept of non-negative matrix factorization (NMF) to polarized sources.
It relies on two key ingredients: \emph{(i)} the algebraic representation of polarization information using quaternions and \emph{(ii)}
the exploitation of physical constraints on Stokes parameters, which generalize non-negativity to the case of polarized signals.
Non-negativity is to be understood here in a \emph{vectorial} sense: it encodes the non-negativity, or positive semi-definiteness, of the 2-by-2 complex covariance matrix associated with each polarized source, as will be explained later.

This \emph{generalized} non-negativity constraint on the source factor makes QNMF fundamentally different from NMF.
In fact, the three Stokes parameters describing polarization are not required to be real non-negative: instead, Stokes parameters must lie inside a particular 4D second-order convex cone.
This means that QNMF cannot be simply reformulated as an augmented real NMF problem by stacking Stokes parameters in a taller matrix.
Instead, QNMF preserves the intrinsic structure of Stokes parameters by leveraging their quaternion representation: it allows the associated generalized non-negativity constraints to be handled straightforwardly.

QNMF encodes a quaternion blind source separation problem.
However, unlike quaternion ICA \cite{javidi2011fast,via2010quaternion} which relies on a standard quaternion linear mixing model and assumes statistical independence of the sources, QNMF does not require any statistical model; it defines a quaternion linear mixing model with non-negativity constraints only.
Such non-negativity constraints are essential to physical interpretation: they guarantee the general relevance of QNMF to solve blind unmixing of spectro-polarimetric data, just like NMF offers a standard approach towards hyperspectral image unmixing \cite{ma2014signal}.
More generally, recent works \cite{flamant2018LTI,flamant2018phdthesis,flamant2017time,flamant2017spectral}
have shown that Stokes parameters provide a natural reparameterization for the second-order properties of bivariate signals.
As a result, QNMF defines a new low-rank approximation tool for arbitrary bivariate signals, such as current velocities in oceanography \cite{gonella1972rotary} or polarized gravitational waves \cite{flamant18eusipco}.

Our contributions are threefold.
First, we introduce the QNMF problem from a spectro-polarimetric linear mixing model perspective and demonstrate how it generalizes the usual NMF.
Second, we study uniqueness conditions for the QNMF problem and reveal the key disambiguating role played by polarization.
QNMF is shown to encompass known NMF sufficient uniqueness conditions, while relaxing NMF necessary uniqueness conditions.
Finally, we provide a simple yet efficient algorithm in the quaternion domain that effectively solves the QNMF problem.
We believe that quaternions, beyond offering an elegant way to handle polarization diversity, permit numerous theoretical (\emph{e.g.} uniqueness conditions) and methodological (\emph{e.g.} algorithmic design) developments that would have been cumbersome to obtain otherwise.

This paper is organized as follows.
Section \ref{sec:preliminaries} collects necessary material regarding quaternions, polarization and associated constraints.
Section \ref{sec:qnmf} defines the concept of quaternion non-negative matrix factorization (QNMF).
Section \ref{sec:identifiabilityNMF} studies uniqueness conditions for QNMF, starting with the simpler 2 sources case, followed by a discussion of the general case.
Section \ref{sec:algorithm} introduces the quaternion alternating least squares (QALS) algorithm for solving the QNMF problem in practice.
Section \ref{sec:numericalExperiments} provides numerical evidence to the relevance of the approach, while Section \ref{sec:conclusion} presents concluding remarks.
Appendices gather technical details and proofs.

\section{Preliminaries}
\label{sec:preliminaries}
This paper relies on the algebraic representation of polarization using quaternions together with the exploitation of structural non-negativity constraints.
To this end,  Section \ref{sub:quaternions} briefly reviews quaternion algebra and Section \ref{sub:polarizationStokes} describes polarization using Stokes parameters.
Section \ref{sub:quaternionPSD} then introduces the quaternion representation of Stokes parameters, whereas Section \ref{sub:polarizationConstraint} details structural non-negative like constraints on Stokes parameters.

\subsection{Quaternions}
\label{sub:quaternions}
Quaternions $\bbH$ forms a 4-dimensional normed division algebra over the real numbers $\bbR$
with canonical basis $\lbrace 1, \bmi, \bmj, \bmk\rbrace$ where $\bmi, \bmj, \bmk$ are imaginary units $\bmi^2 = \bmj^2 = \bmk^2 = -1$ such that $\bmi\bmj = \bmk, \: \bmi\bmj = - \bmj\bmi$ and $\bmi\bmj\bmk = -1$.
Importantly, quaternion multiplication is noncommutative, \ie for $q, q' \in \bbH$, $qq'\neq q'q$ in general.
Any quaternion $q \in \bbH$ reads in Cartesian form $q = a + b\bmi + c\bmj + d\bmk$ where $a, b, c, d \in \bbR$.
The real or scalar part of $q$ is $\real q = a$ and its imaginary or vector part is $\imag q = b\bmi + c\bmj + d\bmk$.
We note $\imagi q = b$, $\imagj q = c$ and $\imagk q = d$.
When $\real q = 0$, $q$ is said to be \emph{pure imaginary}, or simply \emph{pure}.
The quaternion conjugate of $q$ is denoted by $\overline{q} = \real q - \imag q$.
The modulus of $q$ is $\vert q \vert = \sqrt{q\overline{q}} = \sqrt{\overline{q}q}$.

A quaternion matrix $\bfQ \in \bbH^{M\times N}$ has elements $\bfQ_{mn} = q_{mn} \in \bbH$.
Notation $\overline{\bfQ}$ denotes the entry-wise conjugate of $\bfQ$ whereas its conjugate-transpose is denoted by $\bfQ^\dagger = \overline{\bfQ}^\top$.
Usual matrix operations require special care due to non-commutativity of quaternion multiplication: \emph{e.g.} given two matrices $\bfQ_1 \in \bbH^{M\times N}$ and $\bfQ_2 \in \bbH^{N\times P}$, in general one has $\overline{\bfQ_1\bfQ_2} \neq \overline{\bfQ}_1\:\overline{\bfQ}_2 $ and $ (\bfQ_1\bfQ_2)^\top \neq \bfQ_2^\top \bfQ_1^\top$.
However, $(\bfQ_1\bfQ_2)^\dagger = \bfQ_2^\dagger \bfQ_1^\dagger$ is always true.
More details on quaternions and their properties can be found in \cite{conway2003quaternions}. See also \cite{brenner1951matrices,lee1948eigenvalues,zhang1997quaternions,lebihan2004singular,rodman2014topics} and references therein for more on quaternion matrix algebra.

\subsection{Polarization and Stokes parameters}
\label{sub:polarizationStokes}
For electromagnetic (EM) waves such as light traveling in free space, the instantaneous electric field vector $\bfE(t)$ lives in the 2D plane orthogonal to the direction of propagation.
In such transverse waves, the end tip of $\bfE(t) \in \bbR^2$ describes a two-dimensional elliptical trajectory whose properties can vary with time, space or frequency.
This phenomenon is called \emph{polarization}: it describes the geometric nature of wave oscillations.
This behaviour is not limited to electromagnetic radiation: transverse elastic waves and gravitational waves also exhibit polarization.

In optics \cite{born2000principles,brosseau1998fundamentals,J.J.Gil2016}, the polarization state of light is commonly described by four Stokes parameters $S_0, S_1, S_2, S_3$.
These real-valued parameters are energetic quantities that can be measured in experiments \cite{schaefer2007measuring,berry1977measurement}.
The first Stokes parameter $S_0 \geq 0$ is classical and measures the \emph{total} intensity of light, \ie the sum of intensities from the \emph{polarized} and \emph{unpolarized} parts of light.
The three remaining Stokes parameters $S_1, S_2, S_3$ describe the polarized part of light, \ie the properties of the polarization ellipse encoded by the trajectory of the electric field vector.
The \emph{degree of polarization} $\Phi \in [0, 1]$ rules the relative contribution of the polarized part to the total intensity:
\begin{equation}\label{eq:defDegreeOfPolar}
  \Phi = \frac{\text{intensity of polarized part}}{\text{total intensity}} = \frac{\sqrt{S_1^2 + S_2^2 + S_3^2}}{S_0}\:.
\end{equation}
When $\Phi = 1$, light is said to be \emph{fully polarized}, whereas for $\Phi=0$ it is said \emph{unpolarized}. For $\Phi \in (0, 1)$, light is said to be \emph{partially polarized}.

\subsection{Quaternion representation of Stokes parameters}
\label{sub:quaternionPSD}

The Stokes vector $(S_0, S_1, S_2, S_3)^\top \in \bbR^4$ can be conveniently expressed using quaternion algebra as
\begin{equation}\label{eq:quaternionPSDStokes}
w = S_0 + \bmi S_3+ \bmj S_1 + \bmk S_2\: \in \bbH\:.
\end{equation}
This algebraic representation of Stokes parameters has been long standing in the field of optics \cite{Tudor2010a,Tudor2010b,Whitney1971,Pellat-finet1984,Richartz1949}.
More recently \cite{flamant2017spectral,flamant2018LTI}, it has been shown that \eqref{eq:quaternionPSDStokes} exhibits a general relevance for arbitrary bivariate signals.
Our choice for the ordering of Stokes parameters in \eqref{eq:quaternionPSDStokes} adopts that of \cite{flamant2017spectral}.

The quaternion representation of Stokes parameters enables straightforward physical and geometric interpretations while simplifying computations \cite{flamant2018LTI}.
Eq. \eqref{eq:quaternionPSDStokes} can be rewritten for convenience as
\begin{equation}\label{eq:quaternionPSDrwting}
w = I + I\Phi\bmmu
\end{equation}
where $I := S_0 \geq 0$ is the total intensity and $\bmmu$ is the \emph{polarization axis} -- a pure unit quaternion such that $\bmmu^2=-1$.
With this notation, the quantity $\Phi\bmmu$ is a pure quaternion that can be identified with a vector of $\bbR^3$ inside the Poincar\'e sphere \cite{born2000principles} which completely encodes the polarization state of light.
The expression \eqref{eq:quaternionPSDrwting} also emphasizes one of the benefits of the quaternion formalism: it allows a straightforward separation between pure energetic information, conveyed by $\real w = I$, and the geometric vector information encoded by $\imag w = I\Phi\bmmu$.

\subsection{Polarization constraint and quaternion non-negativity}
\label{sub:polarizationConstraint}
Stokes parameters are energetic quantities which are subject to admissibility conditions.
A vector $(S_0, S_1, S_2, S_3)^\top \in \bbR^4$ is admissible as a Stokes vector if it satisfies the following constraints $(\calS)$
\begin{equation}\tag{\calS}\label{eq:polarizationConstraint}
  S_0 \geq 0\quad \text{ and }\quad S_1^2 + S_2^2 + S_3^2 \leq S_0^2\:.
\end{equation}
Those two constraints bear a strong physical interpretation.
The first one is classical and indicates that the total intensity $S_0$ is a non-negative real quantity.
The second one means that the intensity of the polarized part cannot be larger than the total intensity.
As a result, these constraints equivalently encode the range $0 \leq \Phi \leq 1$ of the degree of polarization \eqref{eq:defDegreeOfPolar}.
Note that \eqref{eq:polarizationConstraint} does not impose at all that $S_1, S_2$ and $S_3$ are non-negative.

From a mathematical perspective, \eqref{eq:polarizationConstraint} defines a closed second-order convex cone in $\bbR^4$, denoted by $\bbR^4_{\calS}$.
Conditions \eqref{eq:polarizationConstraint} are referred to as \emph{non-negativity} constraints for polarized signals by the following argument.
Given $(S_0, S_1, S_2, S_3)^\transp \in \bbR^4$, consider the 2-by-2 Hermitian matrix\footnote{In optics, the matrix $\bfJ$ is known as the \emph{polarization} or \emph{coherency} matrix \cite[Section 1.4]{J.J.Gil2016} associated with the bivariate electric field $\bfE$. It can be statistically defined as a covariance matrix \cite[Section 1.4.1]{J.J.Gil2016}, from which Hermitian positive semi-definiteness follows.}:
\begin{equation}
  \bfJ = \frac{1}{2}
  \begin{bmatrix}
    S_0 + S_1 & S_2 + \bmi S_3\\
    S_2 - \bmi S_3 &   S_0 - S_1
  \end{bmatrix}\: \in \bbC^{2\times 2}.
\end{equation}
This mapping is bijective.
Then, a quick check shows that imposing positive semi-definiteness of $\bfJ$ -- or simply, \emph{non-negativity} of $\bfJ$, is equivalent to \eqref{eq:polarizationConstraint}, \ie
\begin{equation}
  (\calS) \Leftrightarrow \begin{cases}
    \trace \bfJ &\geq 0\\
    \det \bfJ & \geq 0
\end{cases}\:.
\end{equation}
This shows that $\bbR^4_{\calS}$ and the set of 2-by-2 Hermitian non-negative matrices are isomorphic.
This result will prove to be essential later on for imposing $\eqref{eq:polarizationConstraint}$ in practice, see Section \ref{sub:projectionConstraints}.
Since the proposed approach relies on the quaternion representation \eqref{eq:quaternionPSDStokes} of Stokes parameters, we define by extension the set of \emph{non-negative quaternions} $\bbHS \subset \bbH$ such that
\begin{equation}
  \bbHS \triangleq  \left\lbrace q \in \bbH \middle \vert \real q \geq 0 \text{ and } \vert \imag q \vert^2 \leq (\real q)^2 \right\rbrace\:.
\end{equation}
By definition, the sets $\bbHS$ and $\bbR^4_{\calS}$ are isomorphic.
They share the same properties:
in particular, this means that $\bbHS$ is a closed second-order convex cone of $\bbH$.

\section{Quaternion Non-negative Matrix Factorization}
\label{sec:qnmf}
This section introduces Quaternion Non-negative Matrix Factorization (QNMF), a new tool which generalizes the notion of Non-negative Matrix Factorization (NMF) to the case of polarized signals.
It exploits two key features: \emph{(i)} the polarization constraint \eqref{eq:polarizationConstraint}, which extends the notion of non-negativity to polarized signals and \emph{(ii)} the algebraic representation \eqref{eq:quaternionPSDStokes} of Stokes parameters using quaternions.
QNMF establishes a new and generic tool for low-rank approximations of polarized signals, with many potential applications in source separation or data completion.

\subsection{Definition}

Without loss of generality, let us consider the setting of a typical spectro-polarimetric experiment.
In such experiment, the four Stokes parameters $S_i(\nu_m, u_n)$, $i=0, 1, 2, 3$ are acquired for a collection of sampled frequencies (or wavelengths) $(\nu_m)_{1\leq m \leq M}$ and spatial locations $(u_n)_{1\leq n \leq N}$.
Using \eqref{eq:quaternionPSDStokes}, the collected data can be written in quaternion form as
\begin{equation}
  \begin{split}
      x(\nu_m, u_n) &= S_0(\nu_m, u_n) + \bmi S_3(\nu_m, u_n) \\
      &+ \bmj S_1(\nu_m, u_n) + \bmk S_2(\nu_m, u_n) \in \bbHS
  \end{split}\label{eq:SpectropolarimetricData}
\end{equation}
for $1\leq m \leq M$, $1\leq n \leq N$.
Note that spectro-polarimetric imaging includes standard hyperspectral imaging as a special case, in which only intensity measurements are performed \ie $x(\nu_m, u_n) = S_0(\nu_m, u_n)$ .

A standard assumption in hyperspectral imaging is to consider that, at a given pixel $u_n$, data can be represented as a linearly weighted combination of $P$ elementary sources (or endmembers) \cite{bioucas2012hyperspectral,ma2014signal}.
While being simple, the linear mixing model appears reasonable for many real-world settings \cite{keshava2002spectral,dobigeon2016linear}.
Generalizing this model to the case of polarized sources, we decompose spectro-polarimetric data \eqref{eq:SpectropolarimetricData} as the linear superposition of $P$ elementary polarized sources $w_p(\cdot) \in \bbHS$:
\begin{equation}\label{eq:mixingModel}
  x(\nu_m, u_n) = \sum_{p = 1}^P w_p(\nu_m)h_p(u_n)\:.
\end{equation}
The quantity $h_p(\cdot) \geq 0$ denotes the spatial activation coefficients attached to the source $p$.
Considering that $h_p(\cdot)$ is a non-negative real quantity has two key implications. First, it means that $h_p(u_n)$ can be interpreted as the abundance of the source $p$ at location $u_n$.
Second, it ensures that, for any elementary source $w_p(\cdot) \in \bbHS$ the linear combination $\sum_{p = 1}^P w_p(\nu_m)h_p(u_n)$ is in $\bbHS$ since $\bbHS$ is a convex cone.

Rewriting \eqref{eq:mixingModel} into matrix form introduces \emph{Quaternion Non-Negative Matrix Factorization} (QNMF) like
\begin{equation}\label{eq:qnmf}
  \bfX = \bfW\bfH
\end{equation}
where $\bfX \in \bbHS^{M\times N}$ is the \emph{data matrix} with coefficients $(\bfX)_{mn} = x_{mn} = x(\nu_m, u_n)$.
The matrix $\bfW \in \bbHS^{M\times P}$ is the \emph{source matrix} with coefficients $(\bfW)_{mp} = w_{mp} = w_p(\nu_m)$, whereas $\bfH \in \bbR_+^{P\times N}$ is called the \emph{activation matrix} such that $(\bfH)_{pn} = h_{pn} = h_p(u_n)$.

\subsection{Relation with NMF}
\label{sub:relationNMF}
The QNMF problem  extends the well-known NMF problem to the case of polarized signals.
Compared to NMF, QNMF features a quaternion-valued sources factor $\bfW$ which exploits the constraint \eqref{eq:polarizationConstraint} instead of the usual non-negativity constraint.
Using quaternion algebra rules, the QNMF \eqref{eq:qnmf} can also be rewritten as
\begin{equation}
  \bfX = \bfW\bfH \Leftrightarrow \begin{cases}
    \real\bfX = \left[\real\bfW\right] \bfH & \text{(NMF)} \\
    \imag\bfX = \left[\imag\bfW\right] \bfH & \text{(polarization)}
\end{cases}\:.\label{eq:QNMFcofact}
\end{equation}
Eq. \eqref{eq:QNMFcofact} shows that QNMF can be seen as a co-factorization problem with common activation factor $\bfH$.
The first factorization problem is an usual NMF on the real part of $\bfX$, \ie on intensity data (Stokes parameter $S_0 \geq 0$) only.
The second one corresponds to a factorization problem on the imaginary part of $\bfX$ encoding polarization properties (Stokes parameters $S_1, S_2, S_3$).
These two factorization problems are not independent for two reasons: \emph{(i)} the activation factor $\bfH$ appears in both and \emph{(ii)} for each coefficient $(m,p)$ of the source factor $\bfW$, the constraint \eqref{eq:polarizationConstraint} connects the modulus of the imaginary part (polarization factorization problem) to its real part (NMF problem) by an inequality constraint.

The relationship \eqref{eq:QNMFcofact} provides another illustration of how QNMF extends NMF to account for polarization diversity.
It allows a precise quantification of the role played by the polarization information and its associated constraint \eqref{eq:polarizationConstraint}.
This will be particularly useful for comparing uniqueness conditions of QNMF and NMF, see Section \ref{sec:identifiabilityNMF}.

\section{Uniqueness conditions for QNMF}
\label{sec:identifiabilityNMF}
This section deals with a fundamental question: upon which conditions on the source $\bfW$ and activation $\bfH$ factors is the QNMF $\bfX = \bfW\bfH$ unique?
The \emph{uniqueness or identifiability} question is central in any factorization or decomposition problem, as it is intimately related to the interpretability of the underlying model.
Many authors have thus considered this question for the NMF problem, either by an analytical study of the properties of the NMF model  \cite{moussaoui2005non,donoho2004does,laurberg2008theorems,huang2014non} or by
 the design of suitable identification criteria to guarantee identifiability \cite{fu2018identifiability,lin2015identifiability,fu2015blind,fu2019nonnegative}.
In this paper, we address the QNMF uniqueness question by a theoretical study of the properties of the QNMF model.
To this aim, Section \ref{sub:trivialAmbNMFuniqueness} starts by introducing the intrinsic ambiguities of the QNMF model.
We provide sufficient uniqueness conditions for QNMF, directly inherited from known NMF ones.
To reveal the key role played by polarization information in model identifiability, Section \ref{sub:twoSourcesIdentifiability} studies in detail the $P=2$ sources case and provides a sufficient condition for QNMF uniqueness.
Section \ref{sub:generalCaseIdentifiability} considers the general $P \geq 2$ case and gives a necessary uniqueness condition that relaxes the requirement for sources $\bfW$ to have some zero entries.

\subsection{Trivial ambiguities and NMF-based uniqueness}
\label{sub:trivialAmbNMFuniqueness}
Let $\bfX \in \bbHS^{M\times N}$ and suppose that there exists two matrices $\bfW \in \bbHS^{M\times P}$ and $\bfH \in \bbR_+^{P\times N}$ such that $\bfX = \bfW\bfH$ holds.
Given a nonsingular matrix $\bfT \in \matrixspace{\bbH}{P}{P}$, any pair $(\tilde{\bfW}, \tilde{\bfH})$ defined as
\begin{align}
  \tilde{\bfW} &= \bfW\bfT \label{eq:indeterminacySourceT}\\
  \tilde{\bfH} &= \bfT^{-1}\bfH \label{eq:indeterminacyActivationT}
\end{align}
leaves the data matrix unchanged \ie $\bfX = \bfW\bfH = \tilde{\bfW}\tilde{\bfH}$.
However, to be admissible, the linear transformation $\bfT$ should yield matrices $\tilde{\bfW}$ and $\tilde{\bfH}$ such that
\begin{equation}\label{eq:constraintTilde}
  \tilde{\bfW} \in \bbHS^{M\times P} \text{ and } \tilde{\bfH} \in \bbR_+^{P\times N}\:.
\end{equation}
The real-valuedness constraint on $\tilde{\bfH}$ directly imposes that $\bfT$ is a real matrix.
Indeed, if the entries $\bfT$ are quaternion-valued, so are the entries of $\bfT^{-1}$, and thus $\tilde{\bfH} = \bfT^{-1}\bfH$ is also quaternion-valued.
As a result, only linear transformations $\bfT \in \bbR^{P\times P}$ may yield alternative factors  $\tilde{\bfW}$ and $\tilde{\bfH}$ satisfying \eqref{eq:constraintTilde}, just as in the standard NMF case.
As a result, QNMF exhibits the same trivial ambiguities as NMF, summarized by Proposition \ref{prop:trivialAmbiguitues} below.
\begin{proposition}\label{prop:trivialAmbiguitues}
The QNMF $\bfX = \bfW\bfH$ exhibits two intrinsic ambiguities, namely,
\begin{itemize}
  \item scale indeterminacy:
  \begin{equation}\label{eq:scaleIndeterminacy}
    \bfT = \mathrm{diag}\left(t_1, t_2, \ldots, t_P\right)\:, t_i > 0, 1 \leq i \leq P
  \end{equation}
  \item order indeterminacy:
  \begin{equation}\label{eq:orderIndeterminacy}
    \bfT \text{ is a permutation matrix}
  \end{equation}
\end{itemize}
  for which $\tilde{\bfW} = \bfW\bfT$ and $\tilde{\bfH} = \bfT^{-1}\bfH$ define equally valid QNMF factors.
\end{proposition}

These inevitable indeterminacies motivate the following definition for the identifiability of the QNMF problem.
\begin{definition}
The QNMF $\bfX = \bfW\bfH$ is said to be essentially unique if $\bfX = \tilde{\bfW}\tilde{\bfH}$ implies $\tilde{\bfW} = \bfW\bfD \bfP$ and $\tilde{\bfH} = \left(\bfD \bfP\right)^{-1}\bfH$ where $\bfD$ is a diagonal matrix with strictly positive entries and $\bfP$ is a permutation matrix.
\label{def:uniquenessQNMF}
\end{definition}
QNMF features the same trivial ambiguities as NMF: those can be handled by standard techniques, \emph{e.g.} by imposing normalization on the source coefficients to get rid of the scale indeterminacy.

As explained in Section \ref{sub:relationNMF}, QNMF generalizes NMF to the case of polarized signals.
The co-factorization perspective given by \eqref{eq:QNMFcofact} shows that QNMF inherits sufficient uniqueness conditions from NMF.
\begin{proposition}[NMF-based sufficient condition for uniqueness]
If the NMF $\real \bfX = \left[\real\bfW\right] \bfH$ is essentially unique, then the QNMF $\bfX = \bfW\bfH$ is essentially unique.
  \label{prop:sufficientConditionGeneralCase}
\end{proposition}
\begin{proof}
  Suppose that $\real \bfX = \left[\real\bfW\right] \bfH$ is essentially unique.
  Thus $\bfH$ can always be written as $\bfH = \bfD \bfP\tilde{\bfH}$  where $\bfD$ is a diagonal matrix with strictly positive entries and $\bfP$ is a permutation matrix.
  Then $\bfX = \bfW\bfH = \tilde{\bfW}\tilde{\bfH}$, with $\tilde{\bfW} = \bfW\bfD \bfP$ and $\tilde{\bfH} = \left(\bfD \bfP\right)^{-1}\bfH$ which shows that the QNMF is essentially unique by Definition \ref{def:uniquenessQNMF}.
\end{proof}
Proposition \ref{prop:sufficientConditionGeneralCase} shows that known sufficient conditions for NMF directly extend to the QNMF case.
For instance, if $\real\bfW$ and $\bfH$ satisfy the separability related assumptions of \cite{laurberg2008theorems,donoho2004does} or the more relaxed sufficiently scattered condition \cite{huang2014non}, then the QNMF $\bfX = \bfW\bfH$ is essentially unique.
While this result is reassuring -- QNMF encompasses known NMF uniqueness conditions -- it does not take advantage of polarization information contained in $\imag \bfW$.
The two remaining sections thus explore how this supplementary information together with the polarization constraint \eqref{eq:polarizationConstraint} impact QNMF model identifiability.

\subsection{Two sources case: range of admissible solutions and sufficient uniqueness condition}
\label{sub:twoSourcesIdentifiability}
We first consider the simpler case of two sources ($P = 2$), which is particularly instrumental in understanding the role played by the polarization constraint \eqref{eq:polarizationConstraint}.
The approach follows closely the one presented in \cite{moussaoui2005non} for the NMF case.

For $1\leq m \leq M$ and $1\leq n\leq N$, the corresponding entries of each of the factors $\bfW$ and $\bfH$ are given by
\begin{equation}
(\bfW)_{mp} = w_{mp} \in \bbHS \text{ and } (\bfH)_{pn} = h_{pn} \geq 0
\end{equation}
where $p=1,2$ denotes the source index.
We use the convenient quaternion form \eqref{eq:quaternionPSDrwting} for $w_{mp}$, \ie
\begin{equation}\label{eq:wmpRewriting}
  w_{mp} = I_{mp} + I_{mp}\Phi_{mp}\bmmu_{mp}\:.
\end{equation}
To model QNMF indeterminacies, we parameterize the matrix $\bfT$  in Eqs. \eqref{eq:indeterminacySourceT} -- \eqref{eq:indeterminacyActivationT} as $\bfT = \bfT(\alpha, \beta)$, $\alpha, \beta \in \bbR$ such that
\begin{equation}\label{eq:TmatrixExplicit}
  \bfT(\alpha, \beta) = \begin{bmatrix}
    1-\alpha & \beta \\
    \alpha & 1-\beta
\end{bmatrix}\:.
\end{equation}
This form for $\bfT(\alpha, \beta)$, together with the assumption $\alpha + \beta < 1$ ensures that $\bfT(\alpha, \beta)$ models only non-trivial  ambiguities of QNMF, \ie other than those given in Proposition \ref{prop:trivialAmbiguitues}.
It follows that $\bfT$ is invertible, with inverse
\begin{equation}\label{eq:TmatrixExplicitInv}
  \bfT^{-1}(\alpha, \beta) = \frac{1}{1-\alpha-\beta}\begin{bmatrix}
    1-\beta & -\beta \\
    -\alpha & 1-\alpha
\end{bmatrix}\:.
\end{equation}

Now consider the linear transformed factors $\tilde{\bfW} = \bfW\bfT(\alpha, \beta)$ and $ \tilde{\bfH} = \bfT^{-1}(\alpha, \beta)\bfH$.
Using \eqref{eq:TmatrixExplicit}, the mixing of sources explicitly reads for $1 \leq m \leq M$
\begin{align}
  \tilde{w}_{m1} & = (1-\alpha)w_{m1} + \alpha w_{m2}\label{eq:mixingSourcesAlpha}\\
  \tilde{w}_{m2} & = \beta w_{m1} + (1-\beta) w_{m2}\label{eq:mixingSourcesBeta}\:.
\end{align}
With \eqref{eq:TmatrixExplicitInv} one gets the new activations for $1 \leq n \leq N$
\begin{align}
  \tilde{h}_{1n} & = \frac{(1-\beta)h_{1n} - \beta h_{2n}}{1-\alpha-\beta} \label{eq:mixingActivAlpha}\\
  \tilde{h}_{2n} & = \frac{-\alpha h_{1n} + (1-\alpha) h_{2n}}{1-\alpha-\beta} \label{eq:mixingActivBeta}\:.
\end{align}
Non-negativity of $\tilde{\bfH}$ imposes that $\tilde{h}_{1n} \geq 0$ and $\tilde{h}_{2n} \geq 0$ for every $n$.
The polarization constraints \eqref{eq:polarizationConstraint} impose that
$\tilde{w}_{m1}, \tilde{w}_{m2} \in \bbHS$.
In other words, for every $m$, one has
\begin{align}
  \text{(NMF)} & \begin{cases}
    \real\tilde{w}_{m1} &\geq 0\\
    \real\tilde{w}_{m2} &\geq 0
  \end{cases} \:,\label{eq:wtildeNMF}\\
  \text{(polarization)} &
  \begin{cases}
      \vert \imag \:\tilde{w}_{m1} \vert^2 \leq \left(\real \:\tilde{w}_{m1}\right)^2\\
        \vert \imag \:\tilde{w}_{m2} \vert^2 \leq \left(\real \:\tilde{w}_{m2}\right)^2
  \end{cases}\label{eq:wtildePol}\:.
\end{align}
Eq. \eqref{eq:wtildeNMF} corresponds also to the usual NMF case \cite{moussaoui2005non}.
However, condition \eqref{eq:wtildePol} reveals the specifity of QNMF.

In order to find the range of admissible solutions for QNMF, we search for intervals $\intqnmf_\alpha$ and $\intqnmf_\beta$ such that, for every $(\alpha, \beta) \in \intqnmf_\alpha \times \intqnmf_\beta$ the linear transform $\bfT(\alpha, \beta)$ yields valid $\tilde{\bfW}$ and $\tilde{\bfH}$ factors.
Remark that, since Eqs. \eqref{eq:mixingActivAlpha} -- \eqref{eq:wtildeNMF} are identical to the NMF case, we can express the intervals $\intqnmf_\alpha$ and $\intqnmf_\beta$ as
\begin{align}
  \intqnmf_\alpha &= \intnmf_\alpha  \bigcap_{m}\intpol_{\alpha, m}\label{eq:intQMFrwting1}\\
  \intqnmf_\beta &= \intnmf_\beta  \bigcap_{m}\intpol_{\beta, m}\label{eq:intQMFrwting2}
\end{align}
where $\intnmf_\alpha$ (resp. $\intnmf_\beta$) is the interval obtained for $\alpha$ (resp. $\beta$) using non-negativity constraints only, as explained below.
For every $m$, $\intpol_{\alpha, m}$ and $\intpol_{\beta, m}$ indicate  intervals for $\alpha$ and $\beta$ such that the polarization constraint \eqref{eq:wtildePol} holds.
Thanks to this rewriting, Eqs. \eqref{eq:intQMFrwting1}--\eqref{eq:intQMFrwting2} permit to separate the contributions of the non-negativity and polarization constraints, respectively.
This allows to precisely quantify the role played by polarization constraint in QNMF to improve NMF identifiability.

\textbf{Computation of $\intnmf_\alpha$ and $\intnmf_\beta$.} These intervals encode 2 non-negativity conditions: the real part of the sources and the activation coefficients.
Their expression follow directly from NMF results \cite{moussaoui2005non}, that is:
\begin{equation}
  \intnmf_\alpha = (\alpha_{\text{min}}^{\text{NMF}}, \alpha_{\text{max}}^{\text{NMF}}] \text{ and } \intnmf_\beta = (\beta_{\text{min}}^{\text{NMF}}, \beta_{\text{max}}^{\text{NMF}}]\:.
\end{equation}
Lower bounds depend on the ratio of the real parts of the sources like
\begin{equation}
 \alpha_{\text{min}}^{\text{NMF}} = - \min_{m \in \bbM_1} \frac{I_{m1}}{I_{m2}-I_{m1}}, \:  \beta_{\text{min}}^{\text{NMF}} = - \min_{m \in \bbM_2} \frac{I_{m2}}{I_{m1}-I_{m2}}
\end{equation}
where for convenience $\real w_{mp} = I_{mp}$ as in \eqref{eq:wmpRewriting} and $\bbM_1 = \lbrace m \vert I_{m2} > I_{m1} \rbrace$,  $\bbM_2 = \lbrace m \vert I_{m1} > I_{m2} \rbrace$ are supposed to be non-empty sets.
If $\bbM_1 = \emptyset$ (resp. $\bbM_2 = \emptyset$), one has $\alpha_{\text{min}}^{\text{NMF}} = -\infty$ (resp. $\beta_{\text{min}}^{\text{NMF}} = -\infty$).
Activations coefficients control the upper bounds of $\intnmf_\alpha$ and $\intnmf_\beta$:
\begin{align}\label{eq:NMFupperBounds}
 \alpha_{\text{max}}^{\text{NMF}} =  \min_{n} \frac{h_{2n}}{h_{1n}+ h_{2n}}, \: \beta_{\text{max}}^{\text{NMF}} =  \min_{n} \frac{h_{1n}}{h_{1n}+ h_{2n}}\:.
\end{align}

\textbf{Computation of $\intpol_\alpha$ and $\intpol_\beta$.}
Let us now exploit the polarization condition \eqref{eq:wtildePol}, on which relies QNMF.
Fix $m$ $(1 \leq m \leq M)$ and start with parameter $\alpha$.
Plugging \eqref{eq:wmpRewriting} into \eqref{eq:mixingSourcesAlpha} ones gets:
\begin{align}
  \real\tilde{w}_{m1} &= (1-\alpha)I_{m1}  + \alpha I_{m2}\:,\\
  \imag\tilde{w}_{m1} &= (1-\alpha) I_{m1}\Phi_{m1}\bmmu_{m1} + \alpha I_{m2}\Phi_{m2}\bmmu_{m2}\:.
\end{align}
By developing all terms appearing in the polarization constraint \eqref{eq:wtildePol} and reorganizing the expression as a second-order polynomial in $\alpha$, we get the following inequality
\begin{align}
&\alpha^2\left[I_{m1}^2 (1 - \Phi_{m1}^2) + I_{m2}^2(1-\Phi_{m2}^2) \right.\nonumber\\
& \qquad\qquad \qquad  \left. - 2I_{m1}I_{m2}(1 - \Phi_{m1}\Phi_{m2} \ip{\bmmu_{m1}, \bmmu_{m2}}) \right]\nonumber\\
  +&2\alpha\left[I_{m1}I_{m2}(1 - \Phi_{m1}\Phi_{m2} \ip{\bmmu_{m1}, \bmmu_{m2}}) - I_{m1}^2 (1 - \Phi_{m1}^2)\right] \nonumber\\
  +& I_{m1}^2 (1 - \Phi_{m1}^2) \geq 0 \label{eq:polynomialAlpha}
\end{align}
where $\ip{\bmmu_1, \bmmu_2} = -\real \bmmu_1\bmmu_2$ can be identified with the usual inner product of $\bbR^3$.
Starting from \eqref{eq:mixingSourcesBeta}, the same approach is used for $\beta$, leading to
\begin{align}
 &\beta^2\left[I_{m2}^2 (1 - \Phi_{m2}^2) + I_{m1}^2(1-\Phi_{m1}^2) \right.\nonumber\\
 & \left.\qquad\qquad\qquad   - 2I_{m1}I_{m2}(1 - \Phi_{m1}\Phi_{m2} \ip{\bmmu_{m1}, \bmmu_{m2}}) \right]\nonumber\\
  +&2\beta\left[I_{m1}I_{m2}(1 - \Phi_{m1}\Phi_{m2} \ip{\bmmu_{m1}, \bmmu_{m2}}) - I_{m2}^2 (1 - \Phi_{m2}^2)\right]\nonumber\\
  +&I_{m2}^2 (1 - \Phi_{m2}^2) \geq 0 \:.\label{eq:polynomialBeta}
\end{align}
For a given $m$, intervals $\intpol_{\alpha,m}$ and $\intpol_{\beta, m}$ are obtained by solving the corresponding polynomial inequalities \eqref{eq:polynomialAlpha} and \eqref{eq:polynomialBeta}.
It involves finding the roots of the associated second-order polynomials in $\alpha$ and $\beta$, respectively.
An easy but lengthy computation of polynomial discriminants shows that \eqref{eq:polynomialAlpha} and \eqref{eq:polynomialBeta} always admit two real-valued solutions (possibly degenerate).
Unfortunately, the amount of parameters involved (at least 4: $\Phi_{m1}, \Phi_{m2}, \ip{\bmmu_{m1}, \bmmu_{m2}}$ and the ratio $I_{m1}/I_{m2}$) prevents from performing a general theoretical study of the roots behaviour.
Such a study would require to account for numerous particular cases, notably due to the fact that the sign of the second-order term is not constant and can even cancel out for some values of parameters.
Nonetheless, in practice when values of sources parameters are given, Eqs. \eqref{eq:polynomialAlpha} and \eqref{eq:polynomialBeta} can be solved very efficiently numerically to yield desired intervals $\intpol_{\alpha, m}$ and $\intpol_{\beta, m}$.
See Section \ref{sec:numericalExperiments} for numerical illustrations on synthetic and real-world data.

\textbf{Sufficient uniqueness condition.}
Despite the apparent complexity of the polarization conditions \eqref{eq:polynomialAlpha} and \eqref{eq:polynomialBeta}, a simple and interpretable sufficient condition for the essential uniqueness of QNMF with $P = 2$ can be formulated.
\begin{proposition}\label{prop:sufficientCondition}
  If the following conditions are satisfied:
  \begin{itemize}
    \item  $\exists \:m_1, m_2 \in \lbrace 1, 2, \ldots M\rbrace$ s.t.
    \begin{equation}
      \begin{cases}
        \Phi_{m_11} = 1, \: \Phi_{m_12}\bmmu_{m_12}\neq \bmmu_{m_11}\\
        \displaystyle I_{m_11} \geq \frac{1}{2}\frac{1-\Phi_{m_12}^2}{1-\Phi_{m_12}\ip{\bmmu_{m_11}, \bmmu_{m_22}}}I_{m_12}\\
        \Phi_{2}(k_2) = 1, \: \Phi_{m_21}\bmmu_{m_21}\neq \bmmu_{m_22}\\
        \displaystyle I_{m_22} \geq \frac{1}{2}\frac{1-\Phi_{m_21}^2}{1-\Phi_{m_21}\ip{\bmmu_{m_22}, \bmmu_{m_21}}}I_{m_21}
      \end{cases}\tag{C1}\label{eq:CS1}
    \end{equation}
    \item $\exists \:n_1, n_2 \in \lbrace 1, 2, \ldots N\rbrace, n_1 \neq n_2$ s.t.
    \begin{equation}
      \begin{cases}
        h_{1n_1} > 0 \text{ and } h_{2n_1} = 0\\
        h_{2n_2} > 0 \text{ and } h_{1n_2} = 0
      \end{cases}\tag{C2}\label{eq:CS2}
    \end{equation}
  \end{itemize}
  then the QNMF $\bfX = \bfW\bfH$ is essentially unique.
\end{proposition}
\begin{proof}
See Appendix \ref{app:proofCS}.
\end{proof}
On the one hand, condition \eqref{eq:CS2} is identical to the one found for the activation factor in the standard NMF case \cite{moussaoui2005non,laurberg2008theorems,huang2014non}.
On the other hand, condition \eqref{eq:CS1} illustrates the key role played by polarization information.
Compared to the usual NMF sufficient conditions for the 2 sources case \cite{moussaoui2005non,laurberg2008theorems,huang2014non}, it does not require each source to vanish alternatively.
Condition \eqref{eq:CS1} shows that it is sufficient that there exists two indices $m_1, m_2$ such that the first source is fully polarized at $m_1$ (and the second one exhibits a different polarization state at $m_1$) and the second one is fully polarized at $m_2$ (and the first one exhibits a different polarization state at $m_2$), and that respective intensities of each source should be larger than the other one by a factor in $[0, 1]$.

Note that \eqref{eq:CS1} does not require at all $m_1 \neq m_2$.
In fact, when $m_1 = m_2 = m$,  it becomes condition \eqref{eq:CS1p}
\begin{equation}
  \begin{cases}
    \Phi_{m1} = \Phi_{m2} = 1, \: \bmmu_{m1} \neq \bmmu_{m2}\\
    I_{m1} > 0, \: I_{m2} > 0
  \end{cases}\:.\tag{CS1'}\label{eq:CS1p}
\end{equation}
In other words, if there exist $m$ such that both sources exhibit some arbitrary energy and are fully polarized with different polarization axes and if \eqref{eq:CS2} is satisfied,  then the QNMF $\bfX = \bfW\bfH$ is essentially unique.

The sufficient condition given in Proposition \ref{prop:sufficientCondition} for the uniqueness of QNMF in the two source case appears remarkably broad compared to the NMF case \cite{moussaoui2005non}.
By incorporating polarization information and its associated constraints, QNMF permits to achieve uniqueness even when the sources never vanish -- a typical case where NMF is known to be not unique \cite{huang2014non,moussaoui2005non,fu2019nonnegative}.

\subsection{General case $(P \geq 2)$}
\label{sub:generalCaseIdentifiability}

The study of the uniqueness of QNMF $\bfX = \bfW\bfH$ for an arbitrary number of $P$ sources is much more cumbersome than the $P=2$ sources case, as the  previous explicit modeling of indeterminacies is no longer amenable.
However, the relationship between QNMF and NMF (see Section \ref{sub:relationNMF}) allows to formulate a necessary condition for uniqueness where standard NMF would fail.

Uniqueness of the NMF $\real \bfX = \left[\real\bfW\right] \bfH$ requires each column of $\real\bfW$ to contain at least one entry equal to zero\footnote{Note that requiring some entries of $\real\bfW$ to be zero is equivalent to require the corresponding entries of $\bfW$ to be zero due to the nature of \eqref{eq:polarizationConstraint}. } \cite{laurberg2008theorems,moussaoui2005non,huang2014non}.
Unfortunately, this condition is often violated, notably in hyperspectral imaging applications, and thus much effort has been devoted to the design of suitable optimization criteria
\cite{craig1994minimum,fu2018identifiability,fu2019nonnegative,lin2015identifiability} to recover uniqueness of the (algorithmic) solution.

In contrast,  QNMF relaxes the NMF zero-entry necessary condition on $\bfW$, stated by Proposition \ref{prop:necessaryCondition} below.
\begin{proposition}[Necessary condition for uniqueness with non-vanishing sources]
  \label{prop:necessaryCondition}
Suppose that the QNMF $\bfX = \bfW\bfH$ is essentially unique such that $\real w_{mp} > 0$ for every $m, p$.
Then the following conditions are satisfied:
\begin{itemize}
  \item $\forall (p, q), \: p \neq q,$
  \begin{equation}
    \exists m \text{ s.t. } \Phi_{mp} = 1, \: \Phi_{mq}\bmmu_{mq} \neq \bmmu_{mp}\tag{A1}\label{eq:CN1}
  \end{equation}
  \item $\forall (p, q), \: p \neq q,$
  \begin{equation}
      \exists n \text{ s.t. } h_{pn} = 0 \text{ and } h_{qn} > 0\tag{A2}\label{eq:CN2}
  \end{equation}
\end{itemize}
\end{proposition}
\begin{proof}
  See Appendix \ref{app:proofCN}.
\end{proof}
Proposition \ref{prop:necessaryCondition} shows that a necessary condition for uniqueness of the QNMF with non-vanishing sources is that, for any distinct pair of sources $(p, q)$, two criteria are satisfied: \eqref{eq:CN1} there exists an index $m$ such that the source $p$ is fully polarized and the source $q$ exhibits a different polarization state (it can be fully polarized but with a different polarization axis); \eqref{eq:CN2} there is an index $n$ such that the source $q$ is active while the source $p$ is not.
Finally, note that condition \eqref{eq:CN2} on the activation factor is exactly the same as in the NMF case \cite{moussaoui2005non,laurberg2008theorems,huang2014non}.

In our opinion, the result stated by Proposition \ref{prop:necessaryCondition} is particularly interesting: it suggests that, in the general case of $P \geq 2$ sources, there may exist sufficient uniqueness conditions for QNMF that do not require at all the source factor $\bfW$ to contains zero-entries, unlike NMF.
This conjecture is motivated by Proposition \ref{prop:sufficientCondition}, which provides such a sufficient condition for the $P=2$ sources case.
It is also supported by our numerical experiments in Section \ref{sub:spectroUnmixingQALS}, where QNMF uniqueness is observed on a synthetic, fully polarized spectro-polarimetric dataset with $P=3$ sources.
QNMF might offer a convenient framework to deal with situations where NMF alone does not guarantee uniqueness, such as with non-vanishing sources in hyperspectral imaging.
This motivates us to search for general sufficient uniqueness conditions in future work.

\section{An algorithm for QNMF}
\label{sec:algorithm}
This section deals with the practical resolution of the QNMF problem.
For a given value $P$, its resolution can be seen as an optimization procedure,
\begin{equation}\label{eq:genericCost}
  \min_{\substack{\bfW \in \bbHS^{M\times P}\\\bfH \in \bbR_+^{P\times N }}} D(\bfX, \bfW\bfH)
\end{equation}
where $D: \bbH^{M\times N} \times \bbH^{M\times N} \rightarrow \bbR_+$ can be an arbitary cost function for quaternion matrices.
For simplicity, we choose here the Frobenius distance between quaternion matrices, referred to as the Euclidean cost in the sequel:
\begin{equation}
  \begin{split}
      D(\bfX, \bfW\bfH) &= \Vert \bfX - \bfW\bfH\Vert_F^2\\
      &= \sum_{m, n} \vert X_{mn} - (\bfW\bfH)_{mn}\vert^2\:.
    \end{split}
\end{equation}
The generic QNMF optimization problem \eqref{eq:genericCost} becomes
\begin{equation}\label{eq:EuclideanCost}
  \min_{\substack{\bfW \in \bbHS^{M\times P}\\\bfH \in \bbR_+^{P\times N }}}  \Vert \bfX - \bfW\bfH \Vert^2_F\:.
\end{equation}
The formulation of the resolution of the QNMF problem \eqref{eq:genericCost} appears much alike the standard NMF with Euclidean cost.
However, two fundamental questions need to be answered: \emph{(i)} is the constraint \eqref{eq:polarizationConstraint} easy to implement? and \emph{(ii)} can we optimize w.r.t. $\bfW$ directly in the quaternion domain?

Fortunately, the answers to these two questions are affirmative.
For \emph{(i)}, the answer relies on the key link between \eqref{eq:polarizationConstraint} and the set of non-negative Hermitian 2-by-2 matrices.
A positive answer to \emph{(ii)} is made possible by the recent advent of the theory of quaternion-domain derivatives \cite{xu2015theory,xu2016optimization,xu2015enabling,mandic2011quaternion} -- the so-called \emph{generalized $\bbH\bbR$ calculus}.

Section \ref{sub:projectionConstraints} below explains how constraints on $\bfW$ and $\bfH$ factors are implemented.
Section \ref{sub:quaternionALS} describes the proposed alternating least squares strategy to solve \eqref{eq:EuclideanCost}.
For completeness, Appendix \ref{app:sub:optimH} provides a quick introduction to quaternion optimization.
Detailed computation of QNMF updates are found in Appendices \ref{app:sub:updateW} and \ref{app:sub:updateH}.

\subsection{Projections onto constraints}
\label{sub:projectionConstraints}

Projections onto the constraint sets of $\bfW$ and $\bfH$ are a cornerstone of any method attempting to solve numerically the QNMF factorization problem \eqref{eq:genericCost} or \eqref{eq:EuclideanCost}.
Projection operators onto constraints $\bbR_+$ and \eqref{eq:polarizationConstraint} are denoted by $\Pi_{\bbR_+}$ and $\Pi_{\bbH_{\calS}}$, respectively.
The projection $\Pi_{\bbR_+}$ is classical \cite{lee1999learning,lee2001algorithms}:
\begin{equation}
  \left[\Pi_{\bbR^+}(\bfM)\right]_{pn} = \max(0, (\bfM)_{pn})
\end{equation}
where $\bfM \in \bbR^{P\times N}$, $1 \leq p \leq P$, $1 \leq n\leq N$.

To compute $\Pi_{\bbH_{\calS}}$, we use the fact that the constraint \eqref{eq:polarizationConstraint} can be derived from non-negativity (\ie semi-definiteness) of 2-by-2 complex  Hermitian matrices -- as described in Section \ref{sub:polarizationConstraint}.
Consider an arbitrary matrix $\bfM \in \bbH^{M\times P}$ with $(m,p)$ coefficient $(\bfM)_{mp} = M_{mp} = \real M_{mp}+ \bmi \imagi M_{mp} + \bmj \imagj M_{mp} + \bmk \imagk M_{mp}$.
Then define the following bijective mapping $f: \bbH \rightarrow \bbC^{2\times2}$ such that
\begin{equation}
  \begin{split}
    &f\left(M_{mp}\right) \\
    &\triangleq
    \frac{1}{2}\begin{bmatrix}
      \real M_{mp} + \imagj M_{mp} & \imagk M_{mp} + \bmi \imagi \:M_{mp}\\
      \imagk M_{mp} - \bmi \imagi \:M_{mp} &   \real M_{mp} - \imagj \:M_{mp}
    \end{bmatrix}
  \end{split}\:.\label{eq:mappingQuatHermitian}
\end{equation}
By construction, $f\left(M_{mp}\right)$ is an Hermitian matrix.
Its projection onto the set of non-negative Hermitian matrices $\calK_+^2$ is well-known \cite[p. 399]{boyd2004convex}, \ie
\begin{equation}
  \Pi_{\calK_+^2}\left[f\left(M_{mp}\right)\right] = \sum_{i=1}^2 \max(0, \eta_i)\bfv_i\bfv_i^\dagger
\end{equation}
where $\eta_i$ and $\bfv_i$ are the $i$-th eigenvalue and eigenvector of the matrix $f\left(M_{mp}\right)$, respectively.
Since $\Pi_{\calK_+^2}\left[f\left(M_{mp}\right)\right]$ is a 2-by-2 non-negative Hermitian matrix, it can be expressed as
\begin{equation}
  \Pi_{\calK_+^2}\left[f\left(M_{mp}\right)\right] = \begin{bmatrix}
    a & c \\
    \overline{c} & b
\end{bmatrix}\label{eq:projectionPSDexplicitMatrix}
\end{equation}
where $a\geq 0$, $b\geq 0$ and $c \in \bbC$ such that $ab-\vert c\vert^2 \geq 0$ on account of positive semi-definiteness of $\Pi_{\calK_+^2}\left[f\left(M_{mp}\right)\right]$.
Remark that \eqref{eq:projectionPSDexplicitMatrix} is of the form \eqref{eq:mappingQuatHermitian}: since the mapping $f$ is bijective, the projection $\left[\Pi_{\bbH_{\calS}}(\bfM)\right]_{mp}$ is uniquely given by
\begin{equation}
  \left[\Pi_{\bbH_{\calS}}(\bfM)\right]_{mp} = a+b + 2\bmi \imag c + \bmj (a-b) + 2\bmk \real c\:,
\end{equation}
where $a, b$ and $c$ as defined above.

\subsection{Proposed algorithm: quaternion alternating least squares}
\label{sub:quaternionALS}

The proposed algorithm adopts a popular strategy for solving the QNMF problem, based on the alternating constrained minimization of \eqref{eq:EuclideanCost} w.r.t. $\bfH$ and $\bfW$.
This choice is motivated by the fact that, whereas the Euclidean cost \eqref{eq:EuclideanCost} is not convex in both $\bfW$ and $\bfH$, it is convex in each variable separately.
After initialization of the factors, the iteration $r > 0$ reads
\begin{align}
  \bfH_{r+1} &\leftarrow \argmin_{\bfH \in \bbR_+^{P\times N }} \Vert \bfX - \bfW_{r}\bfH\Vert^2_F\label{eq:iterationrConstrainedH}\\
  \bfW_{r+1} &\leftarrow \argmin_{\bfW \in \bbHS^{M\times P}} \Vert \bfX - \bfW\bfH_{r+1} \Vert^2_F\label{eq:iterationrConstrainedW}\:.
\end{align}
Equations \eqref{eq:iterationrConstrainedH}-\eqref{eq:iterationrConstrainedW} describe a two-block coordinate descent (CD) scheme in the quaternion domain.
Two-block real-domain CD schemes are known to converge to a stationary point under mild conditions \cite{bertsekas1997nonlinear,grippo2000convergence,kim2014algorithms}.
These convergence results can be readily transposed to the quaternion case by identifying $\bbH$ with $\bbR^4$.
Formally, since $\bbR^+$ and $\bbHS$ are closed convex sets, every limit point of the sequence $(\bfW_r, \bfH_r)$ generated by the two-block CD framework \eqref{eq:iterationrConstrainedH}-\eqref{eq:iterationrConstrainedW} is a stationary point of \eqref{eq:EuclideanCost}.

Convergence guarantees of the two-block CD scheme require to solve exactly both subproblems \eqref{eq:iterationrConstrainedH} and \eqref{eq:iterationrConstrainedW}.
However, the associated constraints make it difficult to obtain these updates directly in a closed form.
The quaternion nature of the data $\bfX$ and source factor $\bfW$ adds a supplementary degree of difficulty due to non-commutativity of quaternion multiplication.
Solving exactly each subproblem \eqref{eq:iterationrConstrainedH}-\eqref{eq:iterationrConstrainedW} requires to properly extend constrained optimization algorithms over the set of quaternions, which is beyond the scope of this paper.

Instead, we propose a simpler, closed-form algorithm that approximately solves each subproblem at each iteration, at the price of losing convergence guarantees.
The main idea is the following: at a given iteration $r$, for each factor, one solves the unconstrained least-squares problem and projects the obtained solution onto the corresponding constraint:
\begin{align}
   \bfH_{r+1} & \leftarrow \Pi_{\bbR_+}\left[\argmin_{\bfH} \Vert \bfX - \bfW_r\bfH\Vert^2_F\right]\\
   \bfW_{r+1} &\leftarrow \Pi_{\bbHS}\left[\argmin_{\bfW} \Vert \bfX - \bfW\bfH_{r+1} \Vert^2_F\right]\:.
\end{align}
Due to its resemblance with the usual alternating least squares (ALS) algorithm \cite{paatero1994positive} for NMF, we call this strategy \emph{quaternion alternating least squares} (QALS).
Derivation of explicit updates for unconstrained least-squares problems requires special care because of the quaternion nature of $\bfX$ and $\bfW$.
It involves the notion of quaternion derivative as well as cautious handling of quaternion non-commutivity, see detailed computations provided in Appendices \ref{app:sub:updateH} and Appendix \ref{app:sub:updateW}.
As a result, one gets the explicit updates:
\begin{align}
  \bfH_{r+1}  &\leftarrow \Pi_{\bbR_+}\left[\left(\real\left[ \bfW^\top_r\overline{\bfW}_r\:\right]\right)^{-1}\real\left[ \bfW^\top_r\overline{\bfX}\:\right]\right]\label{eq:updateHexplicit}\\
  \bfW_{r+1} &\leftarrow \Pi_{\bbHS}\left[ \bfX\bfH_{r+1}^\transp\left(\bfH_{r+1}\bfH^\transp_{r+1}\right)^{-1}\right]\:.\label{eq:updateWexplicit}
\end{align}
Projections onto constraints sets are carried out as described in Section \ref{sub:projectionConstraints}.

The proposed algorithm is remarkably simple and cheap.
Performing optimization directly in the quaternion-domain yields updates expressions that are very much alike the standard ALS algorithm for NMF \cite{berry2007algorithms}.
However, this apparent simplicity should not conceal the important underlying technical details presented in Appendix \ref{app:sec:quaternionALS}.
Such computations would not have been amenable without using the powerful theory of quaternion derivatives introduced recently \cite{xu2015theory,xu2016optimization,xu2015enabling,mandic2011quaternion}.
Despite the lack of theoretical guarantees on its convergence, the quaternion ALS algorithm provides a good baseline for the resolution of QNMF with reasonably good results in most situations.
Thus, it paves the way to further developments of more sophisticated algorithms for the QNMF problem.

\pagebreak
\section{Numerical experiments}
\label{sec:numericalExperiments}

This section illustrates the relevance of the proposed approach by performing numerical experiments on synthetic data.
Section \ref{sub:polarIdentifiability} provides numerical evidence to the key role played by polarization properties onto the range of admissible solutions for the 2 sources case.
Section \ref{sub:spectroUnmixingQALS} demonstrates the effectiveness of the QALS algorithm to solve the QNMF problem.

\subsection{Polarization information and identifiability}
\label{sub:polarIdentifiability}
\begin{figure*}
  \centering
  \includegraphics[width=\textwidth]{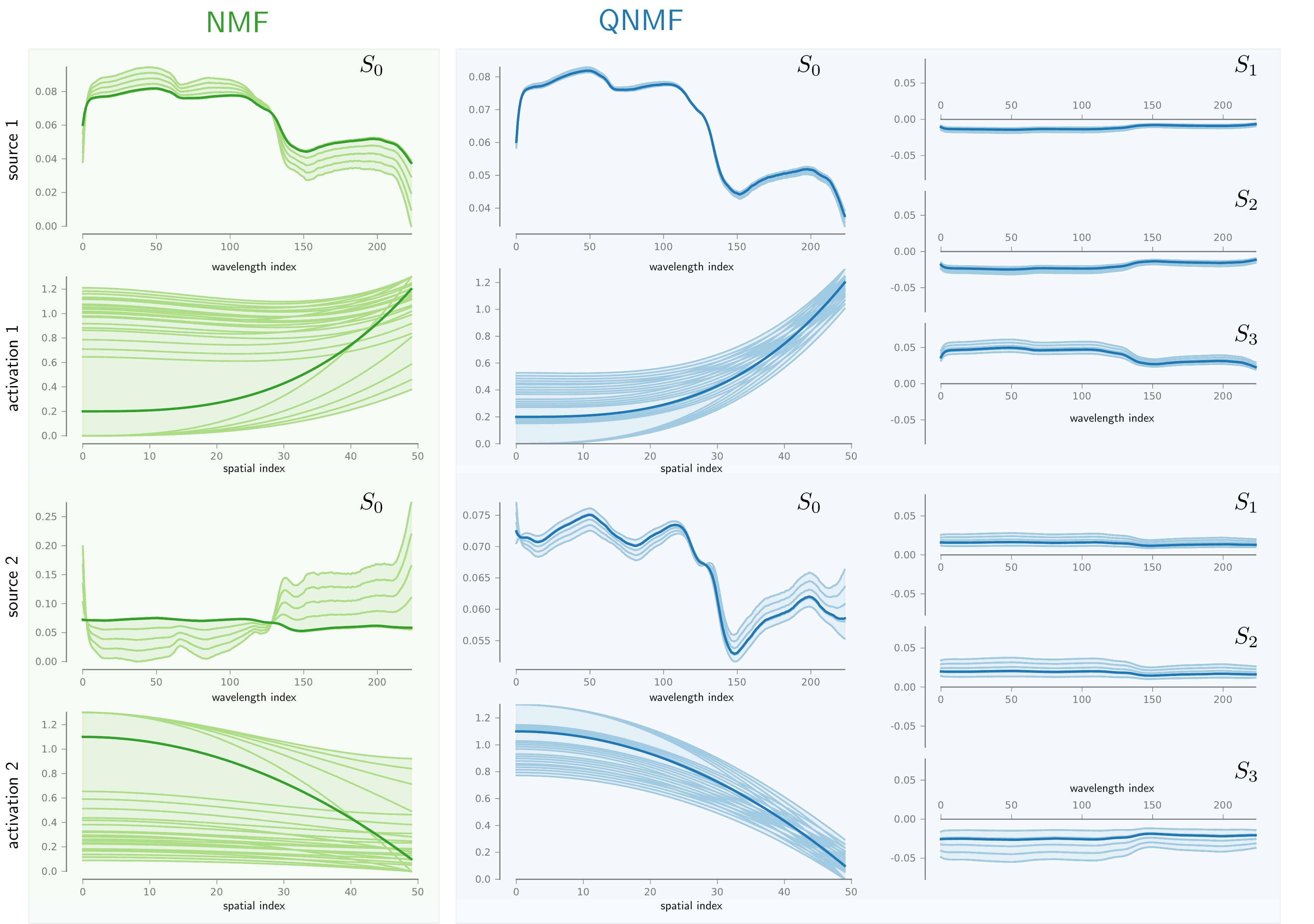}
  \caption{Range of admissible solutions for the QNMF $\bfX = \bfW\bfH$ and the associated NMF $\real\bfX = \left[\real \bfW\right] \bfH$. Taking advantage of polarization information with QNMF significantly improves model identifiability by shrinking the range of admissible solutions, compared to the sole use of intensity ($S_0$) information in the standard NMF approach. }\label{fig:admissibility}
\end{figure*}

This section illustrates the theoretical study provided in Section \ref{sub:twoSourcesIdentifiability} regarding the identifiability of QNMF in the $P=2$ sources case.
To emphasize how polarization impacts the range of admissible solutions, we choose a setting where the QNMF $\bfX = \bfW\bfH$ is not unique -- and thus, neither is the associated NMF $\real\bfX = \real\left[\bfW\right]\bfH$, see \eqref{eq:QNMFcofact}.

Both sources and associated activations are non-vanishing, \ie for all $m,n$, one has $w_{m1}, w_{m2}, h_{1n}, h_{2n} \neq 0$.
Intensity responses $\real w_{m1}, \real w_{m2}$ encode realistic $S_0$ parameters, obtained for instance in hyperspectral imaging of wood components \cite{schwanninger2011review}.
To simplify, both sources are chosen to exhibit constant polarization properties:
\begin{align}
  \Phi_{m1} = 0.7, \quad \bmmu_{m1} =  0.87\bmi - 0.25\bmj - 0.43\bmk\\
  \Phi_{m2} = 0.5, \quad \bmmu_{m2} = -0.71\bmi +  0.44\bmj +  0.55\bmk\:.
\end{align}
Each source being partially polarized, it precludes the QNMF to be unique, since the necessary conditions stated by Proposition \ref{prop:necessaryCondition} are not fulfilled.

Numerical computations yield the range of admissible solutions for the QNMF and its associated NMF problem.
Following Section \ref{sub:twoSourcesIdentifiability}, the range of admissible parameters $\alpha$ and $\beta$ defining the transformation matrix $\bfT(\alpha, \beta)$ are found for the QNMF
\begin{equation}
  \begin{split}
      \intqnmf_\alpha &= [-1.494 \cdot 10^{-1}, 7.692 \cdot 10^{-2}]\\ \intqnmf_\beta &= [-3.719 \cdot 10^{-1},1.538 \cdot 10^{-1}]
  \end{split}
\end{equation}
and for the associated NMF problem:
\begin{equation}
  \begin{split}
      \intnmf_\alpha &= [-1.799, 7.692 \cdot 10^{-2}]\\
      \intnmf_\beta &= [-1.303 \cdot 10^{1},1.538 \cdot 10^{-1}]\:.
  \end{split}
\end{equation}
Note that the QNMF and NMF intervals share the same upper bounds on $\alpha$ and $\beta$, which arise from the activation factor ratios \eqref{eq:NMFupperBounds}.

Fig. \ref{fig:admissibility} represents the set of admissible solutions corresponding to these NMF and QNMF intervals, respectively.
Thick lines indicate the original sources and activations factors.
Comparing respective ranges of solutions for $S_0$, it appears that QNMF significantly improves identifiability over standard NMF, by taking advantage of polarization information.
Improvements are found on both $S_0$ and activations factors, and importantly,
QNMF permits to reconstruct Stokes parameters $S_1, S_2, S_3$ with limited uncertainty.
These results illustrate how QNMF takes advantage of the strong discriminative power of polarization in order to improve model identifiability.

\subsection{Spectro-polarimetric blind unmixing using QALS}
\label{sub:spectroUnmixingQALS}

\begin{figure*}
  \centering
  \includegraphics[width=.8\textwidth]{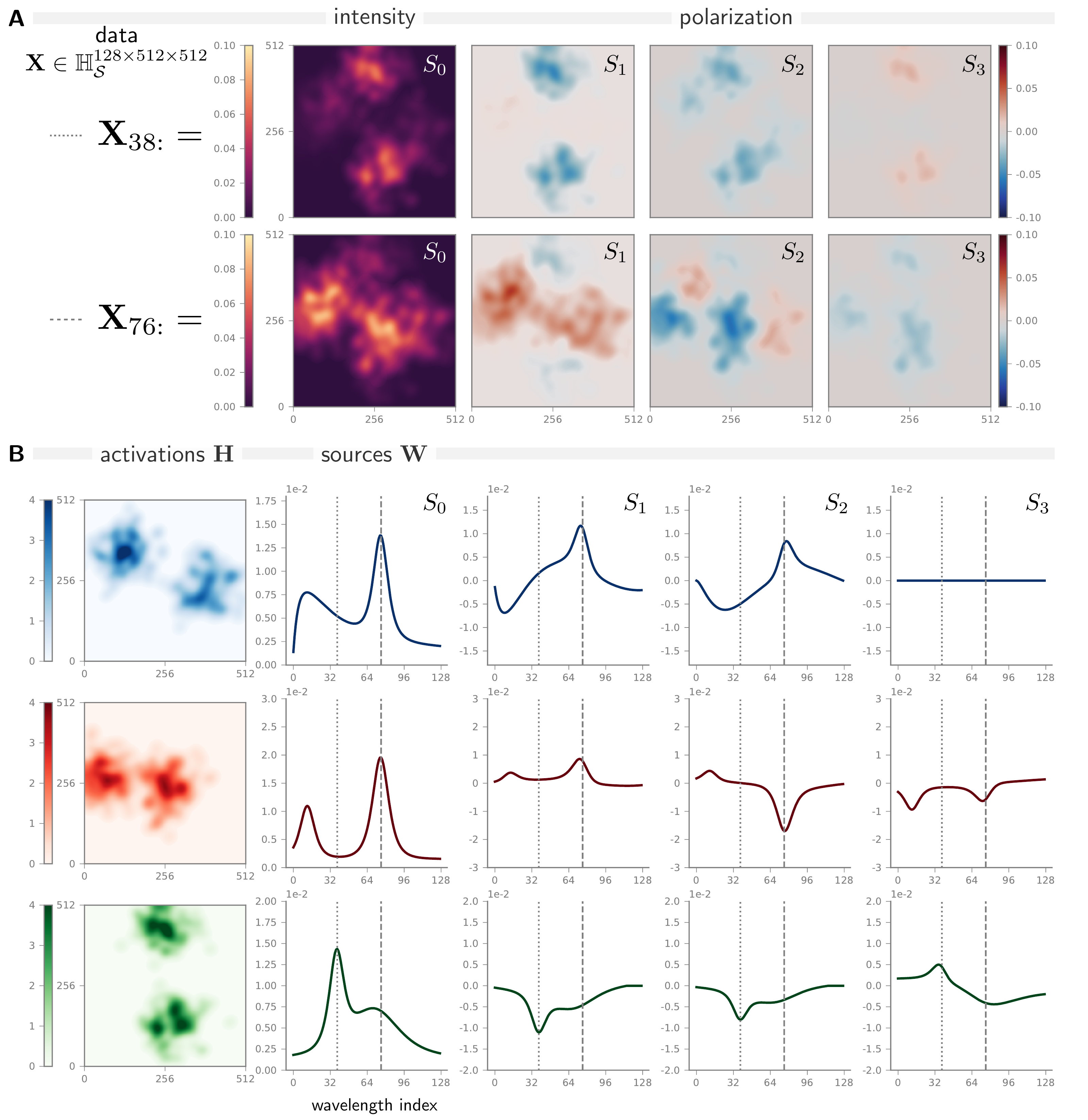}
  \caption{Blind source separation of spectro-polarimetric data using QALS. \textbf{A} 2D intensity ($S_0$) and polarization ($S_1, S_2, S_3$) maps for two wavelength indices $m=38$ and $m=76$, corresponding to local intensity maxima of sources. \textbf{B} reconstructed 2D activations maps and sources Stokes spectral profiles for each of the three factors. Dotted and dashed lines correspond to wavelength indices $m=38$ and $m=76$, respectively. }\label{fig:ALS}
\end{figure*}

This section provides a first numerical validation of the QALS algorithm on realistic simulated spectro-polarimetric data.
It consists in a linear mixture of $P = 3$ sources imaged at $M = 128$ wavelengths and $N = 512 \times 512$ spatial locations.
Each source exhibits broad-band behavior, featuring non-vanishing intensity $S_0$ spectral profile together with spectrally varying polarization $S_1/S_0, S_2/S_0, S_3/S_0$ parameters.
For simplicity, we assume that each source is fully polarized.
Two-dimensional activations maps of each source are designed to represent a typical spectro-polarimetric imaging setting, \emph{e.g.} solar spectro-polarimetry \cite{mccauley2019low,orozco2007quiet}, and such that the resulting activation matrix $\bfH$ satisfies the separability condition \cite{huang2014non}.
As a result, the designed $\bfW$ and $\bfH$ matrices satisfy the necessary conditions for QNMF identifiability stated in Proposition \ref{prop:necessaryCondition}.
Fig. \ref{fig:ALS}A shows simulated spectro-polarimetric data $\bfX = \bfW\bfH$ for two distinct wavelengths indices $m=38$ and $m=76$.

We show that QALS permits to effectively recover the QNMF factors.
We adopt a random initialization strategy.
The initial activation $\bfH_0$ is chosen as a matrix of i.i.d. entries drawn from $\calU([0, 1])$, the uniform distribution on $[0, 1]$.
The initial source factor $\bfW_0$ has i.i.d. entries drawn from the quaternion circular unit Gaussian distribution \cite{le2017geometry} projected onto the constraint \eqref{eq:polarizationConstraint}.
We then alternate Eqs. \eqref{eq:updateHexplicit}--\eqref{eq:updateWexplicit} until convergence.
The latter was assessed by monitoring the relative error $\epsilon_r = \Vert \bfX-\bfW_r\bfH_r\Vert^2_F/\Vert \bfX\Vert^2_F$ and stopping whenever the improvement was below a given threshold, \ie  $\vert \epsilon_r - \epsilon_{r-1} \vert\leq 10^{-5}$ in our case.
Repeating this procedure for $K = 100$ independent initializations of the QALS algorithm, it took on average 29 iterations for the algorithm to converge.
Out of $100$ initializations, we observed that the QALS algorithm recovered each time identical sources and activation factors, strongly indicating that the QNMF of the data $\bfX$ may be essentially unique.

Fig. \ref{fig:ALS}B depicts the 3 sources and corresponding activations factors recovered by the QALS algorithm for one arbitrary initialization.
Excellent match between reconstructed and ground truth factors is indicated by  very small relative errors, \ie $\epsilon_{\bfH} = \Vert \hat{\bfH}-\bfH\Vert^2_F/\Vert \bfH\Vert^2_F = 4.91\cdot 10^{-5}$ and
$\epsilon_{\bfW} = \Vert \hat{\bfW}-\bfW\Vert^2_F/\Vert \bfW\Vert^2_F = 3.92\cdot 10^{-5}$.
Looking at sources intensities $S_0$, we observe that they never vanish, \ie $S_0 > 0$, plus that the first and second sources exhibit very similar peaks.
This is a very challenging situation for NMF-based methods, which would require additional regularizations such as minimum volume \cite{lin2015identifiability,fu2015blind,fu2018identifiability} to ensure uniqueness of the solution.
It is thus remarkable that the QNMF model with the QALS algorithm is able to solve this blind source problem in a simple manner by leveraging polarization information.
One of the advantages of polarization in source separation problems can be directly seen by inspecting $\bfX_{76:}$ in Fig. \ref{fig:ALS}A: while it is not possible to distinguish 2 sources using $S_0$ only, the distinctive polarization pattern on $S_2$ makes it easy to identify source 1 and source 2.

These simulation results demonstrate the ability of the proposed QALS algorithm to effectively solve the QNMF problem.
Despite no associated provable convergence guarantees, QALS offers a simple, computationally efficient algorithm and provides a baseline for the development of more advanced algorithms for QNMF.

\section{Conclusion}
\label{sec:conclusion}

This paper has introduced a new powerful tool called quaternion non-negative matrix factorization (QNMF), which generalizes the well-known concept of non-negative matrix factorization (NMF) to the case of polarized signals.
The algebraic representation of Stokes parameters using quaternions, together with the generalization of the non-negativity constraint on Stokes parameters made it possible to formulate and prove uniqueness results that may have been cumbersome to obtain otherwise.
We unveiled the relation between NMF and QNMF, showing that NMF sufficient conditions are also sufficient uniqueness conditions for QNMF.
Remarkably, we have shown that QNMF improves model identifiability over NMF by leveraging the complete information available in the four Stokes parameters.
This key role played by polarization information is stated by Proposition \ref{prop:sufficientCondition} for the two sources case and by Proposition \ref{prop:necessaryCondition} for the general case.
Furthermore, taking advantage of recent results in quaternion optimization, we have proposed a simple yet efficient algorithm for solving the QNMF problem in practice.
These first results appear very promising and pave the way to future work on both theoretical and methodological aspects of QNMF.
In particular, it is an open question whether one can generalize to the QNMF case the notion of non-negative rank \cite{cohen1993nonnegative} and its related properties.
This fundamental question is intimately related to the definition of a polarized source; as such, it would greatly improve our understanding of the role played by polarization diversity in blind source separation problems.
\appendices

\section{Uniqueness conditions for QNMF}

\subsection{Proof of Proposition \ref{prop:sufficientCondition}}
\label{app:proofCS}

Suppose that $P=2$ and that the QNMF $\bfX = \bfW\bfH$ exists.
We suppose that conditions \eqref{eq:CS1} and \eqref{eq:CS2} are satisfied.
To show that the QNMF is identifiable, we prove that intervals $\intqnmf_\alpha$ and $\intqnmf_\beta$ defined in \eqref{eq:intQMFrwting1}-\eqref{eq:intQMFrwting2} are restricted to $\lbrace 0 \rbrace$.

Condition \eqref{eq:CS1} permits to simplify \eqref{eq:polynomialAlpha} and \eqref{eq:polynomialBeta} for $m_1, m_2$ like
\begin{align}
  &\alpha^2\left[I_{m_12}^2(1-\Phi_{m_12}^2) - 2I_{m_11}I_{m_12}(1 - \Phi_{m_12} \ip{\bmmu_{m_11}, \bmmu_{m_12}}) \right]\nonumber\\
    &+2\alpha\left[I_{m_11}I_{m_12}(1 - \Phi_{m_12} \ip{\bmmu_{m_11}, \bmmu_{m_12}})\right]  \geq 0\label{eq:polynomialAlphaSimp}\\
  &\beta^2\left[I_{m_21}^2(1-\Phi_{m_21}^2) - 2I_{m_21}I_{m_22}(1 - \Phi_{m_21} \ip{\bmmu_{m_21}, \bmmu_{m_22}}) \right]\nonumber\\
     +&2\beta\left[I_{m_21}I_{m_22}(1 - \Phi_{m_21}\Phi_{m_22} \ip{\bmmu_{m_21}, \bmmu_{m_22}})\right] \geq 0\:.\label{eq:polynomialBetaSimp}
\end{align}
Using inequalities linking $I_{m1}, I_{m2}$ for $m=m_1,m_2$ in \eqref{eq:CS1} yields the following domains of solutions to \eqref{eq:polynomialAlphaSimp} and \eqref{eq:polynomialBetaSimp}:
\begin{equation}
  \intpol_{\alpha,m_1} = [0, \alpha_0] \quad   \intpol_{\beta, m_2} = [0, \beta_0],
\end{equation}
where $\alpha_0, \beta_0 \geq 1$.
Condition \eqref{eq:CS2} imply that
\begin{equation}
  \intnmf_\alpha = [\alpha_0', 0] \quad \intnmf_\beta = [\beta_0', 0],
\end{equation}
where $\alpha_0', \beta_0' \leq 0$.
By intersection of intervals, one gets that $\intqnmf_\alpha$ and $\intqnmf_\beta$ defined in \eqref{eq:intQMFrwting1}-\eqref{eq:intQMFrwting2} are restricted to $\lbrace 0 \rbrace$, so that the QNMF is unique.

\subsection{Proof of Proposition \ref{prop:necessaryCondition}}
\label{app:proofCN}
Suppose that the QNMF $\bfX = \bfW\bfH$ exists and that sources never vanish, \ie $\real w_{mp} > 0$ for every $m, p$.
We obtain conditions \eqref{eq:CN1} and \eqref{eq:CN2} by contradiction: one at a time we suppose that either \eqref{eq:CN1} or \eqref{eq:CN2} is false and show that there exists a non-trivial matrix $\bfT$ leading to different factors $\tilde{\bfW} = \bfW\bfT$ and $\tilde{\bfH} = \bfT^{-1}\bfH$.

Suppose that \eqref{eq:CN1} is not satisfied.
Then, $\exists (p_0, q_0), p_0 \neq q_0$ such that
\begin{equation}
  \forall m, \: \left(\Phi_{mp_0} \in [0, 1) \text{ or } \Phi_{mq_0}\bmmu_{mq_0} = \bmmu_{mp_0}\right)\:,
\end{equation}
or equivalently,
\begin{equation}\label{eq:app:codntionNegationRWTING}
  \forall m, \: \left[\Phi_{mp_0} \in [0, 1) \text{ or } \left(\Phi_{mq_0} = 1 \text{ and } \bmmu_{mq_0} = \bmmu_{mp_0}\right) \right]\:,
\end{equation}
Now consider the transformation $\bfT_{p_0q_0}^\alpha \in \bbR^{P\times P}$ defined by
\begin{equation}
    \forall k, \ell \in \lbrace 1, 2, \ldots, P\rbrace,
    \begin{cases}
        \left(\bfT_{p_0q_0}^\alpha\right)_{kk} &= 1\\
        \left(\bfT_{p_0q_0}^\alpha\right)_{p_0q_0} &= -\alpha\\
        \left(\bfT_{p_0q_0}^\alpha\right)_{k\ell} &= 0 \text{ otherwise }\\
    \end{cases}\:.\label{eq:defTpq}
\end{equation}
By construction $\bfT_{p_0q_0}^\alpha$ does not correspond to a trivial ambiguity of the QNMF when $\alpha \neq 0$.
Note that $(\bfT_{p_0q_0}^\alpha)^{-1} = \bfT_{p_0q_0}^{-\alpha}$.
Consider the new factors $\tilde{\bfW} = \bfW\bfT_{p_0q_0}^\alpha$ and $\tilde{\bfH} = \bfT_{p_0q_0}^{-\alpha}\bfH$ such that $\bfX = \tilde{\bfW}\tilde{\bfH}$.
From \eqref{eq:defTpq}, the source indexed by $p_0$ is the only one affected by $\bfT_{p_0q_0}^\alpha$ like
\begin{equation}
  \forall m,\: \tilde{w}_{mp_0} = w_{mp_0} - \alpha w_{mq_0}\:.
\end{equation}
The corresponding activation coefficients read
\begin{equation}
  \forall n, \tilde{h}_{p_0n} = h_{p_0n} + \alpha h_{q_0n}\:.\label{eq:app:activCoeffCN}
\end{equation}
Supposing $\alpha > 0$, then $\tilde{h}_{p_0n} \geq 0$ for every $n$ by non-negativity of the matrix $\bfH$.
It remains to find at least one $\alpha > 0$ such that $\forall m,\: \tilde{w}_{mp_0} = w_{mp_0} - \alpha w_{mq_0}\in \bbHS$.
Imposing $\tilde{w}_{mp_0} \in \bbHS$ for every $m$, yields
\begin{equation}
  \forall m, \begin{cases}
    I_{mp_0} - \alpha I_{mq_0} & \geq 0\\
    \left\vert I_{mp_0}\Phi_{mp_0}\bmmu_{mp_0} -\alpha I_{mq_0}\Phi_{mq_0}\bmmu_{mq_0}\right\vert^2 &\leq \left(I_{mp_0} - \alpha I_{mq_0} \right)^2
  \end{cases}\label{eq;NN+polapp}
\end{equation}
where we used notation \eqref{eq:wmpRewriting} for convenience.
Since by assumption $I_{mp_0}, I_{q_0} > 0$, the first condition implies that $\alpha \leq \min_m (I_{mp_0}/I_{mq_0})$.
The second condition related to polarization can be rewritten as, for every $m$
\begin{equation}
  \begin{split}
    & \alpha^2I_{mq_0}^2(1-\Phi_{mq_0}^2) \\
      &- 2\alpha I_{mp_0}I_{mq_0}(1-\Phi_{mp_0}\Phi_{mq_0}\ip{\bmmu_{mp_0},\bmmu_{mq_0}})\\
      &+I_{mp_0}^2(1-\Phi_{mp_0}^2) \geq 0 \:.
  \end{split}\label{eq:app:polynomialCN}
\end{equation}
According to \eqref{eq:app:codntionNegationRWTING} two cases may occur, not exclusively from each other.
Fix $m$ and suppose that $\Phi_{mp_0} \in [0, 1)$.
Two cases are possible: either $\Phi_{mq_0} = 1$ or $\Phi_{mq_0} \in[0, 1)$.
Assume that $\Phi_{mq_0} = 1$, then \eqref{eq:app:polynomialCN} is true for $\alpha \leq I_{mp_0}(1-\Phi_{mp_0}^2)I_{mq_0}^{-1}((1-\Phi_{mp_0}\ip{\bmmu_{mp_0},\bmmu_{mq_0}})^{-1}$.
For $\Phi_{mq_0} \in [0, 1)$, since $I_{mq_0}^2(1-\Phi_{mq_0}^2) >0$ in virtue of assumptions, the discriminant $\Delta$ reads
\begin{equation}
  \begin{split}
      \Delta &= 4I_{mp_0}^2I_{mq_0}^2\left[(1-\Phi_{mp_0}\Phi_{mq_0}\ip{\bmmu_{mp_0},\bmmu_{mq_0}})^2\right.\\
      &\left. - (1-\Phi_{mp_0}^2)(1-\Phi_{mq_0}^2)\right]\:.
  \end{split}
\end{equation}
Ranges of parameters imply that $\Delta \geq 0$.
As a result, \eqref{eq:app:polynomialCN} is true for $\alpha \in (-\infty, \alpha_-^m] \cup [\alpha_+^m, +\infty)$ where $\alpha_{\pm}^m$ are the roots of the polynomial.
Sign of polynomial coefficients in \eqref{eq:app:polynomialCN} together with $0\leq \Phi_{mp_0} < 1$ imply that $\alpha_{\pm}^m > 0$.
Suppose now that $\Phi_{mq_0} = 1$ and $\bmmu_{mq_0} = \bmmu_{mp_0}$.
Taking $\Phi_{mp_0} = 1$ Eq. \eqref{eq:app:polynomialCN} becomes trivial. The case $\Phi_{mq_0} \in [0, 1)$ is included in the previous discussion.

Summarizing all cases, for every $m$ such that $\Phi_{mp_0} \in [0, 1)$ or $\Phi_{mq_0}\bmmu_{mq_0} = \bmmu_{mp_0}$, there always exists $\alpha > 0$ such that conditions \eqref{eq;NN+polapp} are satisfied, meaning that the QNMF is not unique.
This leads to the first necessary condition \eqref{eq:CN1}.
Repeating the same approach with $\alpha < 0$ yields condition \eqref{eq:CN2}.

\section{Derivation of quaternion ALS updates}
\label{app:sec:quaternionALS}
\subsection{Quaternion derivatives using generalized $\bbH\bbR$-calculus}
\label{app:sub:optimH}

Given a function $f :\bbH \rightarrow \bbH$ of the variable $q\in \bbH$, one outstanding question is: does its derivative $\partial f/\partial q$ exists and if so, how do we compute it?
Quaternion analytic functions \cite{sudbery1979quaternionic,watson2003generalized} are known to be differentiable, unfortunately they form a very restricted class of functions -- of little interest to signal processing.
In fact, cost functions $f :\bbH \rightarrow \bbR$ are not analytic \cite{sudbery1979quaternionic}, so that other strategies need to be deployed.
First, a pseudo-derivative approach can be used by treating $f(q)$ as a function $f(q_0, q_1, q_2, q_3)$ of the four real components of $q$.
As pointed out in \cite{xu2015theory}, such approach requires lengthy and cumbersome computations, thus limiting its applicatibility.
In order to compute derivatives directly in the quaternion domain, crucial steps have been made recently with the development of the $\bbH\bbR$-calculus \cite{mandic2011quaternion}, and subsequently the generalized $\bbH\bbR$-calculus \cite{xu2015enabling}.
The latter provides a complete framework including chain rule and product rules, extending naturally the $\bbC\bbR$-calculus \cite{kreutz2009complex} of complex-valued optimization to the case of quaternion functions.

For the quadratic-type of quaternion functions encountered in this paper, only important results related to quaternion optimization are necessary.
These are given in Proposition \ref{prop:stationaryPoint} and Table \ref{table:quaternionFunctionGradient} below.
We refer to  \cite{xu2015enabling,xu2015theory} for a complete description of the generalized $\bbH\bbR$-calculus.

Proposition \ref{prop:stationaryPoint} characterizes stationary points of cost functions of quaternion matrices.
\begin{proposition}[Theorem 4.1 \cite{xu2015theory}]\label{prop:stationaryPoint}
  Let $f: \bbH^{M\times N}\rightarrow \bbR$ and denote by $\nabla_{\bfQ}f$ its gradient w.r.t. $\bfQ$. Then,
  \begin{equation}
    \begin{split}
        \bfQ_0 \text{ is a stationary point } \\
        \Leftrightarrow \nabla_{\bfQ}f(\bfQ_0) = 0 \Leftrightarrow  \nabla_{\overline{\bfQ}}f(\bfQ_0) = 0
    \end{split}\:\:.
  \end{equation}
\end{proposition}
For practical computations,  Table \ref{table:quaternionFunctionGradient} reprints from \cite{xu2015theory} some useful functions of the quaternion matrix variable $\bfQ$ and their corresponding gradients.

\subsection{Updates for $\bfH$}
\label{app:sub:updateH}
Due to non-commutativity of the quaternion product, standard rules matrix derivatives cannot be applied directly.
Instead, let us write explicitly the Euclidean cost function
\begin{align}
  \Vert \bfX -  \bfW\bfH\Vert^2_F &= \sum_{m,n} \left\vert x_{mn} - \sum_k w_{mk}h_{kn}\right\vert^2\\
  & \hspace{-7em}= \sum_{m,n}\left( x_{mn} - \sum_k w_{mk}h_{kn}\right)\left( \overline{x_{mn}} \sum_k \overline{w_{mk}} h_{kn}\right)\:,\label{eq:euclideanCostExplicit}
\end{align}
where we have used that $\overline{w_{mk} h_{kn}} = \overline{w_{mk}} h_{kn}$ since $h_{kn}$ is real-valued.
Using the chain rule, the partial derivative of \eqref{eq:euclideanCostExplicit} w.r.t. $h_{ij}$ is
\begin{align}
  \frac{\partial \Vert \bfX -  \bfW\bfH\Vert^2_F}{\partial h_{ij}} & \notag\\
  &\hspace{-6em}= - \sum_{m,n}\frac{\partial \sum_k w_{mk}h_{kn} }{\partial h_{ij}} \left( \overline{x_{mn}}- \sum_k \overline{w_{mk}} h_{kn}\right) \notag\\
  &\hspace{-6em}\phantom{=} - \sum_{m,n} \left( w_{mn} - \sum_k w_{mk}h_{kn}\right)\frac{\partial \sum_k \overline{w_{mk}} h_{kn} }{\partial h_{ij}}\\
  &\hspace{-6em}= -\sum_m w_{mi} \left( \overline{x_{mj}} - \sum_k \overline{w_{mk}}h_{kj}\right)\notag\\
  &\hspace{-6em}\phantom{=} - \sum_m \left( x_{mj} - \sum_k w_{mk}h_{kj}\right)\overline{w_{mi}}\\
  &\hspace{-6em}= -\sum_m w_{mi} \left( \overline{x_{mj}} - \overline{(\bfW\bfH)_{mj}}\right)\notag \\
  &\hspace{-6em}\phantom{=}- \sum_m \left(x_{mj} - (\bfW\bfH)_{mj}\right)\overline{w_{mi}}\\
  &\hspace{-6em}= -2 \real\left[ \sum_m w_{mi} \left(\overline{x_{mj}} -\overline{(\bfW\bfH)_{mj}}\right)\right]\\
  &\hspace{-6em}= -2\real\left[ \bfW^\top\overline{\left(\bfX - \bfW\bfH\right)}\right]_{ij}\:.
\end{align}
As a result, one gets
\begin{equation}\label{eq:gradientHexplicit}
  \nabla_{\bfH} \Vert \bfX - \bfW\bfH\Vert^2_F = -2\real\left[ \bfW^\top\overline{\left(\bfX - \bfW\bfH\right)}\right]\:.
\end{equation}
The Euclidean cost function is convex in $\bfH$; therefore, the minimizer ${\bfH}_\star$ of this cost is obtained by cancelling out the gradient \eqref{eq:gradientHexplicit}:
\begin{equation}
  -2\real\left[ \bfW^\top\overline{\bfX}\right]  + 2\real\left[ \bfW^\top\overline{\bfW}{\bfH}_\star\right] = 0\:,
\end{equation}
where again we used that $\overline{\bfW{\bfH}_\star} = \overline{\bfW}{\bfH}_\star$ since ${\bfH}_\star$ is a real matrix.
It also implies that $\real\left[ \bfW^\top\overline{\bfW}{\bfH}_\star\right] = \real\left[ \bfW^\top\overline{\bfW}\right]{\bfH}_\star$, leading to the following expression for ${\bfH}_\star$:
\begin{align}
  {\bfH}_\star &= \argmin_{\bfH} \Vert \bfX - \bfW\bfH\Vert^2_F\\
  &= \left(\real\left[ \bfW^\top\overline{\bfW}\right]\right)^{-1}\real\left[ \bfW^\top\overline{\bfX}\right]\:.
\end{align}
\begin{table}
  \centering
  \caption{Matrix functions $f: \bbH^{M\times N} \rightarrow \bbH$ and corresponding gradients. Matrices $\bfA_1, \bfA_2$ are quaternion-valued with appropriate dimensions.}
  \begin{tabular}{cc}
    \toprule
    $f(\bfQ)$ & $\nabla_{\overline{\bfQ}} f$\\
    \midrule
    $\trace \bfA_1 \bfQ \bfA_2$ & $- \frac{1}{2}\bfA_1^\transp \bfA_2^\dagger$ \\
    $\trace \bfA_1 \bfQ^\dagger \bfA_2$ & $\real[\bfA_2] \bfA_1$\\
    $\trace \bfA_1 \bfQ^\dagger\bfQ \bfA_2$& $\real[\bfQ\bfA_2]\bfA_1 -\frac{1}{2}(\bfA_1\bfQ^\dagger)^\transp\bfA_2^\dagger$\\
    \bottomrule
  \end{tabular}
 \label{table:quaternionFunctionGradient}
\end{table}
\subsection{Updates for $\bfW$}
\label{app:sub:updateW}
Let us rewrite the Euclidean distance \eqref{eq:EuclideanCost} in terms of a trace operator
\begin{equation}\label{eq:EuclideanDistTrace}
  \Vert \bfX - \bfW\bfH\Vert^2_F  = \trace \left[\left(\bfX - \bfW\bfH\right)^\dagger\left(\bfX - \bfW\bfH\right)\right],
\end{equation}
where $\bfQ^\dagger = \overline{\bfQ}^\top$ is the (quaternion) conjugate-transpose of $\bfQ$.
Developping\footnote{We use that, for any quaternion matrices $\bfA$ and $\bfB$, $(\bfA\bfB)^\dagger = \bfB^\dagger\bfA^\dagger$. See e.g. \cite{zhang1997quaternions} for more complete properties of quaternion matrices.} \eqref{eq:EuclideanDistTrace} one gets
\begin{equation}\label{eq:EuclideanDistTraceDevlopped}
  \begin{split}
        \Vert \bfX - \bfW\bfH\Vert^2_F &= \trace \bfX^\dagger \bfX - \trace \bfX^\dagger\bfW\bfH \\
        &- \trace \bfH^\dagger\bfW^\dagger \bfX + \trace \bfH^\dagger\bfW^\dagger \bfW\bfH\:.
  \end{split}
\end{equation}
By Proposition \ref{prop:stationaryPoint}, stationary points are defined by
cancelling values of the gradient of \eqref{eq:EuclideanDistTraceDevlopped} w.r.t. to $\overline{\bfW}$.
Thanks to Table \ref{table:quaternionFunctionGradient} one gets
\begin{align}
  \nabla_{\overline{\bfW}} \Vert \bfX -\bfW\bfH\Vert^2_F & = \frac{1}{2}\left(\bfX^\dagger\right)^\transp\bfH^\transp - \real[\bfX]\bfH^\transp   \notag\\
  &\hspace{-5em}+ \real[\bfW\bfH]\bfH^\transp -\frac{1}{2}\left(\bfH^\transp\bfW^\dagger\right)^\transp\bfH^\transp\\
  &\hspace{-5em}= -\left(\real[\bfX] - \frac{1}{2}\overline{\bfX}+ \frac{1}{2}\overline{\bfW\bfH} - \real[\bfW\bfH]\right)\bfH^\transp \\
  &\hspace{-5em}=-\frac{1}{2}\left(\bfX - \bfW\bfH\right)\bfH^\transp\:.\label{eq:gradientWexplicit}
\end{align}
The Euclidean cost function is convex in $\bfW$.
Thus, cancelling out the gradient \eqref{eq:gradientWexplicit} yields the global minimizer ${\bfW}_\star$ of this cost such that
\begin{equation}
  -\bfX\bfH^\transp  + {\bfW_\star}\bfH\bfH^\transp = 0
\end{equation}
so that
\begin{align}
  {\bfW}_\star &= \argmin_{\bfW} \Vert \bfX - \bfW\bfH\Vert^2_F  =  \bfX\bfH^\transp\left(\bfH\bfH^\transp\right)^{-1}\:.
\end{align}

\bibliographystyle{IEEEtran}
\bibliography{refs}

\end{document}